\documentclass[10pt,conference]{IEEEtran}

\usepackage{subcaption} 

\usepackage{amsfonts}
\usepackage{mathtools}

\usepackage{zi4}        

\usepackage{listings}   
\usepackage{xcolor}     

\usepackage{xpatch}
\usepackage{algpseudocode}
\usepackage{algorithm}
\usepackage{booktabs}
\usepackage{centernot}
\usepackage{tikz}
\usepackage{pgfplots}
\usepackage{bussproofs}
\usepackage{colortbl}

\usepackage{enumitem}
\usepackage{hyperref}
\usepackage{multirow}

\usepackage[defaultcolor=red, final]{changes}

\newcommand\copyrighttext{%
  \footnotesize \textcopyright 2025 IEEE.  Personal use of this material is permitted.  Permission from IEEE must be obtained for all other uses, in any current or future media, including reprinting/republishing this material for advertising or promotional purposes, creating new collective works, for resale or redistribution to servers or lists, or reuse of any copyrighted component of this work in other works.
}
\newcommand\copyrightnotice{%
\begin{tikzpicture}[remember picture,overlay]
\node[anchor=south,yshift=10pt] at (current page.south) {\fbox{\parbox{\dimexpr\textwidth-\fboxsep-\fboxrule\relax}{\copyrighttext}}};
\end{tikzpicture}%
}

\newcommand{\nocompile}[1]{}
\newcommand{\mkregion}[1]{\langle#1\rangle}
\newcommand{\ptrdomain}[2]{{#2}^\mathcal{#1}}
\newcommand{\memread}[2]{\ast\mkregion{#1,#2}}
\newcommand{\memreadS}[2]{#1.\!\ast\!\!\mkregion{#2}}

\newcommand{\reg}[1]{\texttt{\textcolor{blue}{#1}}}

\newcommand{\mkinitial}[1]{#1\ensuremath{_{\text{\tiny{\fontfamily{zi4}\selectfont{0}}}}}}
\newcommand{\regz}[1]{\text{\mkinitial{\reg{#1}}}}
\newcommand{\hexa}[1]{\ensuremath{\mathtt{#1}}}
\newcommand{\htriple}[3]{\{~~#1~~\}\mbox{\hspace{2ex}}#2\mbox{\hspace{2ex}}\{~~#3~~\}}
\newcommand{\grayhline}{\arrayrulecolor{gray}\hline\arrayrulecolor{black}}

\newcommand{\abstraction}{\alpha_\mathbb{E}}

\newcommand{\tainted}{\top}
\newcommand{\mkabstract}[1]{\overline{#1}}

\definecolor{lightgreen}{RGB}{153,255,51}

\newcommand{\STAB}[1]{\begin{tabular}{@{}c@{}}#1\end{tabular}}

\DeclareMathOperator{\step}{step}

\makeatletter
\def\hlinewd#1{%
\noalign{\ifnum0=`}\fi\hrule \@height #1 \futurelet
\reserved@a\@xhline}
\makeatother

\algrenewcommand\textproc{}

\newtheorem{example}{Example}
\newtheorem{definition}{Definition}
\newtheorem{theorem}{Theorem}

\setlist{leftmargin=*}

\usepackage{xcolor}
\lstdefinestyle{custom}{
    basicstyle=\ttfamily,       
    keywordstyle=\color{blue},  
    escapeinside={(*@}{@*)},    
    breaklines=true             
}
\renewcommand{\lstinline}[1]{%
  \text{\lstinline[style=custom]{#1}}%
}

\begin{document}
\title{Formally Verified Binary-level Pointer Analysis}

\author{
    \IEEEauthorblockN{Freek Verbeek\IEEEauthorrefmark{1}\IEEEauthorrefmark{2}, Ali Shokri\IEEEauthorrefmark{2}, Daniel Engel\IEEEauthorrefmark{1}, Binoy Ravindran\IEEEauthorrefmark{2}}
    \IEEEauthorblockA{\IEEEauthorrefmark{1}Open Universiteit, Heerlen, The Netherlands\\
    Email: \{freek.verbeek, daniel.engel\}@ou.nl}
    \IEEEauthorblockA{\IEEEauthorrefmark{2}Virginia Tech, Blacksburg, VA, USA\\
    Email: \{freek, ashokri, binoy\}@vt.edu}
}

\maketitle
\copyrightnotice

\begin{abstract}
Binary-level pointer analysis can be of use in symbolic execution, testing, verification, and decompilation of software binaries.
In various such contexts, it is crucial that the result is trustworthy, i.e., it can be formally established that the pointer designations are overapproximative.
This paper presents an approach to formally proven correct binary-level pointer analysis.
A salient property of our approach is that it first generically considers what proof obligations a generic abstract domain for pointer analysis must satisfy.
This allows easy instantiation of different domains, varying in precision, while preserving the correctness of the analysis.
In the trade-off between scalability and precision, such customization allows ``meaningful'' precision (sufficiently precise to ensure basic sanity properties, such as that relevant parts of the stack frame are not overwritten during function execution) while also allowing coarse analysis when pointer computations have become too obfuscated during compilation for sound and accurate bounds analysis.
We experiment with three different abstract domains with high, medium, and low precision.
Evaluation shows that our approach is able to derive designations for memory writes soundly in COTS binaries, in a context-sensitive interprocedural fashion.
\end{abstract}


\section{Introduction}\label{sec:intro}

Pointer analysis is central to various forms of verification and analysis for software containing pointers, facilitating the construction of a state-based semantic model of software~\cite{choi1993efficient,landi1993interprocedural,andersen1994program,emami1994context,ramalingam94,steensgaard1996points}.
It aims to statically resolve, for any pointer in a given program, which region of the memory it may point to.
Specifically, given any two pointers, it must be known whether they are  \textit{aliasing}, always referring to \textit{separate} regions in memory, or  if they may possibly overlap.
If a value is written to memory, and no pointer information is known, then one cannot accurately describe what the next state will be.
This can lead to overapproximative \emph{thrashing} parts of the memory state or \emph{forking}, i.e., to conservatively considering both separation and aliasing possibilities.
Both such cases are undesirable as they will quickly lead to unrealistic states and path explosions.
In other words, pointer analysis is a necessity for building a state-based transition system that accurately models the software under investigation.
Such a model, then, typically precedes a verification or analysis effort aimed at higher-level properties.

The necessity for pointer analysis immensely exacerbates when dealing with binaries (i.e., machine or assembly code) instead of source code.
The reason is that at this level of abstraction \emph{everything is a pointer}.
There are no variables, and memory can be considered as a flat unstructured address space.
In a typical x86-64 program, about 28\% of all assembly instructions write to memory\footnote{Measured over several CoreUtils binaries and  Firefox libraries with different levels of optimization.}, producing tens of thousands of pointers even in medium-sized programs.
Moreover, control-flow related information is stored in writable memory, such as the current return address and the currently caught exception stack.
Theoretically, if the destination of even a single memory write cannot be resolved, then the effect of executing that memory write could result in either thrashing all memory (including such control-flow pertinent information) or forking into unrealistic states, such as when a memory write overlaps with the return address, even though in reality it did not.
This may lead to a situation where it cannot even be established what instruction is to be executed, let alone what an accurate next state can be.
This is one of the key challenges in dealing with binaries, preventing one from simply using techniques developed for source code analysis and applying them to low-level code found in binary executables~\cite{debray1998alias}.

This paper presents an approach to formally verified binary-level pointer analysis.
Typically, such analyses are based on a form of abstract interpretation~\cite{cousot1977abstract}, where an abstract domain is defined that overapproximates concrete semantics~\cite{balakrishnan2004analyzing,navas2012signedness}.
A fundamental challenge is choosing the ``right'' abstract domain, as this essentially boils down to balancing precision vs. scalability.
This paper thus first leaves the abstract domain \emph{polymorphic} and formulates a set of eight generic to-be-refined functions, as well as the proof obligations that these functions must satisfy.
Over these generic functions, an executable algorithm for pointer analysis is formalized and proven correct.
An \emph{instantiation} thus defines an abstract domain and an implementation of the generic functions.
Any instantiation that satisfies the proof obligations automatically constitutes a formally proven correct binary-level pointer analysis. As will be discussed later, all the formalism and proofs are carried out in Isabelle/HOL and have been shared with the readers.

We then provide three different instantiations of our generic functions, each of which strikes a different balance in the trade-off between precision vs. scalability.
First, \emph{pointer computations} form an abstract domain that keeps track of how pointers were computed: highly precise, but in practice one must cap the domain to a given size and produce top ($\top$) when that cap is exceeded.
Second, \emph{pointer bases} form an abstract domain where each pointer is represented only by its pointer base (e.g., a stack pointer, or the return value of a \texttt{malloc}): more coarse as it cannot be used for alias analysis, but still allows accurate separation analysis (for binaries compiled from source code, pointers based in different blocks can be assumed to be separate, even if an out-of-bounds occurs~\cite{leroy2008formal}).
Third, \emph{pointer sources} are an abstract domain where a pointer is modeled by the set of sources (e.g., user inputs, initial parameter values) used in its computation.
This is highly coarse, but is scalable and still allows a form of separation reasoning.

The pointer analysis presented in this paper is context-sensitive and compositional.
Context-sensitivity is desirable, since pointers are passed through from function to function.
We derive \emph{function preconditions}, which states that invariably a certain function is always called within in a context where, e.g., register \reg{rdi} contains a heap-pointer, or register \reg{rsi} contains a pointer to within the stack frame of the caller.
Per function, we can then derive a \emph{function postcondition} that summarizes which regions were written to/read from by the function, and if after termination abstract pointers are left in return registers or global variables.
Compositionally, these summaries can then be used for pointer analysis for callers of summarized functions.
Due to space limitation, we focus our presentation on intraprocedural analysis and do not expand on the above technique for composition.

Bottom-up pointer analysis (i.e., binary analysis, in contrast to top-down source code analysis) can be useful in various use cases (see Section~\ref{sec:usecases}):
\begin{itemize}
\item It can be integrated into a disassembly algorithm~\cite{pang2021sok}.
A large facet of disassembly is assessing which instruction addresses are reachable (i.e., control flow recovery).
A key challenge is resolving indirections, i.e., dynamically computed control flow transfers.
Context-sensitive pointer analysis can assist, by providing information on which pointers are passed to a function.
\item It can be a preliminary step to a decompilation effort~\cite{cifuentes1995decompilation}.
Specifically, one of the steps in decompilation is to recover variables. Bottom-up pointer analysis provides information on which memory writes actually constitute variables.
\item It can be the base of a bottom-up dataflow analysis. The function summaries already provide a form of dataflow analysis, by providing information on in- and output relations. They can be used to verify whether functions adhere to a calling convention and to see which state parts are overwritten or preserved by a given function. We also demonstrate by example that an overapproximative pointer analysis can be used for live variable analysis.
\end{itemize}

\deleted{
We emphasize the need for customizable binary-level pointer analysis in favor of traditional means of pointer analysis, in a bottom-up setting. 
First, stripped assembly code contains no variables, no types, and no (de)reference operators. 
Traditional means of pointer analysis typically implicitly assume that their target language has such code constructs.
Second, pointer computations are hard to ``recover'' from assembly to an analyzable form, as they may be implemented by low-level implementations not recognizable as arithmetic.
Traditional methods of pointer analysis often require overapproximative bounds on memory accesses, implicitly assuming that their target language performs arithmetic pointer computations through recognizable arithmetic operators. 
As an example, consider the following assembly: \mintinline{nasm}{mov edx,0x1999999A;}\mintinline{nasm}{imul eax}.
These two instructions will write the value $\reg{eax}/10$ (32-bit integer division by ten) into the register \reg{edx} (even though the latter instruction performs 64-bit multiplication).
Even if an overapproximative bound on the \emph{value} of the pointer cannot be established, it can still be the case, e.g., that the \emph{base} of the pointer can overapproximatively be established.
In other words, binary-level pointer analysis must be able to deal with scenarios where coarse abstraction is needed.
Still, in various cases, precision is desirable: e.g., when analyzing overlap between local variables and the top of the stack frame where the return address is stored.
}

We \deleted{also} emphasize the need for formally proven correct binary-level pointer analysis.
Symbolic treatment of pointers and memory is notoriously difficult.
Existing approaches typically make various assumptions \emph{implicitly}, e.g., they may implicitly assume a return address cannot be overwritten, assume separation between pointers based on heuristics or best practices, or assume alignment of regions.
This paper aims to reduce the trusted code base by explicitizing such assumptions and either proving them through invariants or reporting them explicitly otherwise.
The trusted code base is thus reduced to the validity of explicit and configurable assumptions such as ``regions based on stack pointers of different functions are assumed to be separate''.

\textbf{Limitations, assumptions and scope.} A major assumption behind our approach is the treatment of \emph{partially overlapping} memory accesses.
Memory accesses (reads or writes) to partially overlapping regions may happen, but we assume that they do not concern pointers.
More details can be found in Section~\ref{sec:csemantics}.
Moreover, our approach has been implemented for the x86-64 architecture and does not deal with concurrency.
A fundamental limitation is that not all indirections may be resolved, which may lead to unexplored paths.

The approach has been formally proven correct in Isabelle/HOL~\cite{nipkow2002isabelle,dawson2009isabelle}, and has been mirrored in Haskell for experimental results.
These confirm soundness relative to a ground truth obtained by observing executions.
Moreover, they show precision comparable to or improved upon the state-of-the-art.
We evaluate the effect of interprocedural bottom-up pointer analysis with respect to resolving indirections, identifying 135 cases where context-sensitive information allowed resolving a function callback.
Finally, for all analyzed functions it has been verified whether the result is sufficiently precise to show that the return address has not been overwritten and that critical parts of the stack frame (e.g., storing non-volatile register values) are unmodified during execution of the function.
This was successful for 99.6\% of all analyzed functions.


In summary, we contribute:
\begin{itemize}[noitemsep,topsep=0pt]
\item A \emph{formally proven correct} approach to binary-level pointer analysis that leaves the abstract domain \emph{generic}, allowing easy development of instantiations with different characteristics (e.g., different levels of preciseness);
\item An evaluation over roughly 1.4 million assembly instructions, showing scalability and applicability of the approach.
\end{itemize}

Section~\ref{sec:related-work} studies the related work. Sections~\ref{sec:overview} and \ref{sec:definitions} provide details of our generic functions and their three different levels of instantiation, respectively. Section~\ref{sec:intra}, demonstrates the realization of interprocedural pointer analysis through our general functions. While Section~\ref{sec:usecases} looks at several use cases of the introduced approach, Section~\ref{sec:evaluation} relates the experimental results to those produced by the state-of-the-art tools currently available. Section~\ref{sec:conclusion} concludes the paper.

\section{Related Work}\label{sec:related-work}

Source-level pointer analysis has been an active research field for decades~\cite{hind2001pointer}.
Its typical use cases lie in \emph{top-down} contexts: it takes as input source code, and provides information to the compiler for doing optimizations and data flow analyses.
We here do not aim to provide an overview of this field, since our work focuses on \emph{bottom-up} contexts: taking as input a binary, the result of pointer analysis provides information usable for decompilation and verification.

Generally, source-level approaches to program analysis cannot directly be applied to binary-level programs~\cite{balakrishnan2010wysinwyx}.
Various research therefore focuses on symbolic execution and/or abstract interpretation specifically tailored to the binary level.
We distinguish \emph{under}approximative techniques from \emph{over}approximative ones.
In this discussion, we specifically focus on how these techniques deal with pointers.

\nocompile{
\begin{table}
\centering
\begin{tabular}{llll}
\hlinewd{1pt}
\textbf{Name} & \textbf{Appr. } & \textbf{Pointer technique} & \textbf{References} \\
\hline
\textsc{SAGE}                & Under   & Concretization        & \cite{GodefroidLM08}\\
\textsc{S$^2$E}              & Under   & Concretization        & \cite{chipounov2012}\\
\textsc{Fuzzball}            & Under   & Concretization        & \cite{martignoni2012}\\
\textsc{Binsec}              & Under   & Forking               & \cite{djoudi2015}\\
\grayhline
\textsc{Jakstab}             & Over    & Thrashing                  & \cite{kinder12}\\
\textsc{Hoare Graphs}        & Over    & Forking                    & \cite{verbeek2022formally}\\
\textsc{BinTrimmer}          & Over    & Value Set Analysis         & \cite{redini2019b}\\\grayhline
\textsc{Codesurfer/x86}      & Over    & Value Set Analysis         & \cite{balakrishnan2005codesurfer}\\
\textsc{BAP}                 & Over    & Value Set Analysis         & \cite{brumley2011}\\
\textsc{angr}                & Over    & Value Set Analysis         & \cite{shoshitaishvili16}\\
\textsc{This paper}   & Over    & Value Set Analysis       & This paper \\
\hlinewd{1pt}
\end{tabular}
\caption{Binary-level analysis tools} 
\label{tab:related}
\end{table}
}

\subsubsection*{Underapproximative approaches}\hspace{1ex}\\
\textsc{SAGE} combines symbolic execution of assembly code with fuzz testing, allowing exposure of real-life vulnerabilities in real-life software~\cite{GodefroidLM08}.
Initially, \textsc{SAGE} did ``not reason about symbolic pointer dereferences'', but it has been combined with Yogi allowing runtime behavior observed from test cases to be used in refining abstractions to find whether pointers may be possibly aliasing~\cite{godefroid2008automating}.
\textsc{S$^2$E} is a platform for traversing binaries, allowing exploration of hundreds of thousands of paths using selective symbolic execution~\cite{chipounov2012,chipounov2009}.
\textsc{S$^2$E} provides an approach where dereferencing a symbolic pointer provides next states with possible concrete values based on the symbolic pointer and the current path constraints.
A similar approach is taken by \textsc{FuzzBall}, a symbolic execution framework for binaries, mainly concerned with improving the path coverage of binary fuzzers~\cite{Babi2011,martignoni2012}.
It explores individual paths one by one and chooses concrete values for offsets in pointers.
\textsc{Binsec} is a code analysis tool with a focus on the security properties of binaries~\cite{djoudi2015,daniel2020binsec}.
It has been applied to find use-after-free bugs~\cite{feist2016} and for reachability analyses~\cite{girol2021}.
Symbolic execution is based on forking, using an SMT solver to prune infeasible paths.
\textsc{Binsec} is based on \emph{bounded} verification, making it underapproximative.
Kapus et al. provide an interesting approach by concretizing and segmenting the memory model so that symbolic pointers can only refer to single memory segments~\cite{kapus2019}.
Use of a test harness ensures both termination and that allocations always have a concrete size.

These methods underapproximate either by not exploring all paths, or by underapproximating pointer values.
Underapproximation typically is very well suited for finding bugs and vulnerabilities in software: it leads to few false negatives and provides excellent scalability.
It is suitable in the context of testing and binary exploration.
In contrast, our pointer analysis is overapproximative, quantifying over all execution paths and all values.
This makes it suitable in the contexts of verification and lifting to higher-level representations.

\subsubsection*{Overapproximative approaches}\hspace{1ex}\\
\textbf{Control Flow Reconstruction.} Overapproximative approaches to binary analysis are generally based on a form of abstract interpretation.
Many approaches are aimed at \emph{control flow reconstruction} and resolving of indirect branches.
By having abstract values represent a set of possible jump targets, an indirection can be resolved by concretizing the abstract value to all instruction addresses it represents.
\textsc{Jakstab} performs binary analysis and control flow reconstruction based on this principle~\cite{kinder2010static,kinder12}. 
The user manually provides a harness modeling the initial state, and relative to that initial state, the generated control flow graph overapproximates all paths in the binary.
Pointer aliasing is dealt with by thrashing the symbolic state~\cite{kinder2008jakstab}. 
Verbeek et al. present an overapproximative approach to control flow reconstruction based on forking the state non-deterministically~\cite{verbeek2022formally}.
Their output can be exported to the Isabelle/HOL theorem prover where it can be formally proven correct.
An important use case for overapproximative control flow reconstruction is \emph{trimming} a binary by removing provably unreachable code.
\textsc{BinTrimmer} uses abstract interpretation to prove the reachability of dead code and trim the binary accordingly.

Using abstract interpretation for the purpose of control flow reconstruction (i.e., resolving of indirections) is different from using abstract interpretation for resolving designations of memory writes (i.e., pointer analysis).
It can be seen as a specific form of the more generic technique \emph{value set analysis} (VSA).


\textbf{Value Set Analysis.}
Various approaches use abstract interpretation to do VSA: mapping state parts to abstract representations of a set of values that the state part may hold at a certain program point.
\textsc{Codesurfer/x86} utilizes as abstract domain a tuple storing a base and an offset~\cite{balakrishnan2004analyzing,reps2008improved}.
The offset is modeled by an abstract domain that combines intervals and congruences.
Frameworks such as \textsc{BAP}~\cite{brumley2011} and \textsc{angr}~\cite{shoshitaishvili16} provide implementations of binary-level VSA.
Abstract domains are typically a form of signed-agnostic intervals~\cite{navas2012signedness}.
To our knowledge, BinPointer~\cite{kim2022binpointer} is the work closest related to the contribution in this paper.
Kim et al. provide binary-level context-insensitive pointer analysis in a sound and overapproximative fashion, while also targeting scalability and evaluating preciseness of their produced output.
Their abstract domain is conceptually equivalent to the ``pointer bases'' domain presented in this paper, which is also similar to the abstract domain found in~\cite{balakrishnan2004analyzing}.
BinPointer reaches the conclusion of 100\% soundness by running test-cases, while we reach 100\% soundness through formal proofs.


\textbf{Summary of relation to overapproximative approaches.}
Existing approaches use abstract interpretation to reconstruct control flow, or to do VSA.
It is a well-known issue that in realistic, optimized and stripped COTS binaries, computations quickly become too complicated and obfuscated to be amenable for VSA sufficiently precise to enable reasoning over separation~\cite{zhang2019bda}.
This paper presents the first approach to VSA that is formally proven correct, that is generic wrt. the abstract domain and that therefore allows different domains with different levels of preciseness to be used.
To the best of our knowledge, there exists no approach that can overapproximatively assign pointer designations to virtually all memory writes in large COTS binaries, in a context-sensitive interprocedural fashion.
In Section~\ref{sec:evaluation} we aim to provide a more technical head-to-head comparison with existing tools.


\section{Overview of Generic Constituents}\label{sec:overview}

There are three major constituents required to formulate the correctness theorem.
First of all, a \emph{concrete semantics} that provides a step function $\step$ over concrete states (denoted by $s$, $s'$, $\ldots$).
Second, an \emph{abstract semantics}, defined by 1.) an abstract step function $\mkabstract{\mbox{step}}$ over abstract states (denoted by $\sigma$, $\sigma'$, $\ldots$), and 2.) a \emph{join} (denoted ${\sqcup}$).
Third, a concretization function $\gamma_\mathbb{S}$ that maps abstract states to sets of concrete states.

We prove the following two theorems:
\begin{theorem}
Function $\gamma_\mathbb{S}$ is a \emph{simulation relation}~\cite{baier2008principles} between the concrete and abstract semantics:
\[
	s \in \gamma_\mathbb{S}(\sigma) \implies \step(s) \in \gamma_\mathbb{S}(\mkabstract{\mbox{step}}(\sigma))
\]
\end{theorem}
This theorem shows that the abstract semantics overapproximate the concrete ones.

Second, we define an algorithm that performs symbolic execution, while maintaining a mapping $\phi$ from the visited instruction addresses to the abstract states.
Whenever an instruction address is visited twice, the current abstract state is joined with the abstract state stored during the previous visit, and the algorithm proceeds if the joined state is unequal to the stored one.
In essence, this is a fixed-point computation.
We prove:

\begin{theorem}
The mapping $\phi$ produced by the algorithm provides \emph{invariants}:
\[
    s~\mbox{is reachable} \implies s \in \gamma_\mathbb{S}(\phi(s.\reg{rip}))
\]
\end{theorem}

In words, any concrete reachable state $s$ is included in the set of states represented by the abstract state stored in mapping $\phi$ associated to its instruction pointer $\reg{rip}$.
Reachability means that $s$ is reachable from some unconstrained concrete state, with \reg{rip} as the entry point, through a path of \emph{resolved} control flow transfers.

\subsection{Concrete Semantics}\label{sec:csemantics}

The concrete semantics are largely straightforward, except for the treatment of partially overlapping memory accesses.
At the binary level, any memory access occurs through pointer computations.
We formulate the assumption that \emph{any region in memory storing a pointer is from there on not accessed in a partially overlapping fashion}.
For example, consider a scenario where region $\mkregion{\regz{rsp} - 16,8}$ has been accessed (denoting the 8-byte region 16 bytes below the original value of the stackpointer $\reg{rsp}$).
From that point on, regions $\mkregion{\regz{rsp} - 16,4}$ and $\mkregion{\regz{rsp} - 12,4}$ are still considered valid accesses.
Region $\mkregion{\regz{rsp} - 12,8}$ is not, as it partially overlaps with a previously accessed region.
As a consequence of this assumption, any partially overlapping access is assumed not to produce a pointer.
Since we are interested in pointer analysis, we therefore allow the concrete semantics to make values read from or written to by partially overlapping access to become \emph{tainted}.
Overapproximation then concerns untainted values \emph{only}.

The concrete states stores \emph{concrete values}, denoted $\mathbb{V}$.
A concrete value is either an immediate bitvector or the special value $\tainted$ (tainted).
A concrete state contains an assignment of registers to concrete values.
Concrete memory assigns concrete values to concrete regions: tuples of type $\mathbb{V} \times \mathbb{V}$ containing the address and size of the region.
A priori, no memory alignment information is available and thus concrete memory must store concrete values as well as the current alignment.
Every read/write updates the alignment information and taints values in memory accordingly.
The result is a type $\mathbb{S}$ modeling \emph{concrete states} and function $\step$: a formal but fully executable semantics, in which small assembly programs can be executed on concrete initial states.

\subsection{Abstract Semantics}

The abstract state stores \emph{abstract values}.
The datatype for abstract values, denoted $\mkabstract{\mathbb{V}}$, is left completely polymorphic. 
We assume existence of a special $\mkabstract{\top}$ element.
Over this datatype, the following generic (i.e., \emph{undefined}) functions are assumed to be available:
\[
\begin{array}{cll}
\gamma_\mathbb{V}  & \mbox{Concretization} & \mkabstract{\mathbb{V}} \mapsto \{\mathbb{V}\}\
\\
\mkabstract{S} & \mbox{Abstract semantics} & \mbox{Operation} \times [\mkabstract{\mathbb{V}}] \mapsto  \mkabstract{\mathbb{V}}
\\
\sqcup  & \mbox{Join} & \mkabstract{\mathbb{V}} \times \mkabstract{\mathbb{V}} \mapsto \mkabstract{\mathbb{V}}
\\
\mkabstract{\bowtie} & \mbox{Separation} & \mkregion{\mkabstract{\mathbb{V}} \times \mkabstract{\mathbb{V}}} \times \mkregion{\mkabstract{\mathbb{V}} \times \mkabstract{\mathbb{V}}} \mapsto \mathbb{B} 
\\
\mkabstract{\sqsubseteq} & \mbox{Enclosure} & \mkregion{\mkabstract{\mathbb{V}} \times \mkabstract{\mathbb{V}}} \times \mkregion{\mkabstract{\mathbb{V}} \times \mkabstract{\mathbb{V}}} \mapsto \mathbb{B} 
\\
\mkabstract{==} & \mbox{Aliassing} & \mkregion{\mkabstract{\mathbb{V}} \times \mkabstract{\mathbb{V}}} \times \mkregion{\mkabstract{\mathbb{V}} \times \mkabstract{\mathbb{V}}} \mapsto \mathbb{B} 
\end{array}
\]
For sake of presentation, we omit functions used for getting initial abstract values.
Function $\mkabstract{S}$ takes as input the name of an operation executed by an assembly instruction (e.g, \reg{add}).
Note that a single assembly instruction may execute several such operations (e.g., \reg{imul}).
Function $\mkabstract{S}$ symbolically executes that operation on abstract values.
The separation, enclosure and aliassing relations are over abstract regions.
Function $\gamma_\mathbb{R}$ concretizes an abstract region $\mkregion{a,\mathit{si}}$ by applying $\gamma_\mathbb{V}$ to both the address and the size.

\begin{figure*}
\begin{tabular}{ll}
The join must be overapproximative:&
$\begin{array}{l}
	\gamma_\mathbb{V}(a_0) \subseteq \gamma_\mathbb{V}(a_0 \sqcup a_1)
\end{array}$
\\
Separation must be overapproximative:&
$\begin{array}{l}
	{r_0} ~\mkabstract{\bowtie}~ {r_1} \wedge r_0 \in \gamma_\mathbb{R}({r_0}) \wedge r_1 \in \gamma_\mathbb{R}({r_1}) \implies r_0 \bowtie r_1
\end{array}$
\\
Semantics must be overapproximative:&
$\begin{array}[t]{l}
	v_0 \in \gamma_\mathbb{V}(a_0) \cup \{\tainted\} \wedge v_1 \in \gamma_\mathbb{V}(a_1) \cup \{\tainted\} \implies 
	v_0~\Box~v_1 \in \gamma_\mathbb{V}(\mkabstract{S}(\mkabstract{\Box}, a_0, a_1)) \cup \{\tainted\}
\end{array}$
\\
The join respects separation:&
$\begin{array}{l}
	{r} ~\mkabstract{\bowtie}~ \mkregion{a_0 \sqcup a_1, \mathit{si}} \implies {r} ~\mkabstract{\bowtie}~ \mkregion{a_0, \mathit{si}}
\end{array}$
\\
The join respects enclosure:&
$\begin{array}{l}
	\mkregion{a_0, \mathit{si}} ~\mkabstract{\sqsubseteq}~ \mkregion{a_0 \sqcup a_1, \mathit{si}}
\end{array}$
\\
Enclosure in separate region:&
$\begin{array}{l}
	{r_0} ~\mkabstract{\sqsubseteq}~ {r_1} \wedge {r_1} ~\mkabstract{\bowtie}~ {r_2} \implies {r_0} ~\mkabstract{\sqsubseteq}~ {r_2}
\end{array}$
\\
\end{tabular}
\caption{Examples of Proof Obligations}
\end{figure*}

The above generic functions must satisfy a set of 26 proof obligations.
Most of these are trivial (e.g., the join must be a commutative semigroup over $\mkabstract{\mathbb{V}}$; enclosure is transitive and reflexive).
We here provide some examples of interesting proof obligations.
Join, separation, enclosure, aliasing, and abstract semantics must be overapproximative.
For separation, this means that concretizing separate abstract regions must produce separate concrete regions; note that concrete separation $\bowtie$ can be expressed in terms of linear equalities for non-tainted values.
For the semantics, overapproximation is formulated by stating that the result of applying some concrete operation $\Box$ must be overapproximated by the semantics provided by $\mkabstract{S}$ for the corresponding symbolic operation $\mkabstract{\Box}$. Here, $\Box$ can, e.g., be an arithmetic, logical or bitvector operation.
The presentation shows binary operators, but this easily extends to n-ary operators.
Finally, there is a set of algebraic properties concerning the relations over regions, such as ``Enclosure in separate region''.

\section{Overview of Instantiations}\label{sec:definitions}

We provide three different instantiations, each of which has been formally defined in Isabelle/HOL and for each of which all proof obligations have been proven.
All three domains are represented by sets of elements from a different universe: $\mathcal{C}$, $\mathcal{B}$ and $\mathcal{S}$.
Elements of these domains satisfy the following syntax:
%
\[
\ptrdomain{C}{\underbracket{\{c_0,c_1,\ldots\}}_{\mbox{\begin{tabular}{c}$\mathcal{C}$onstant\\\hspace{-11ex}\rlap{Computation}\end{tabular}}}}
\mbox{\hspace{3ex}}\mid\mbox{\hspace{3ex}}
\ptrdomain{B}{\underbracket{\{b_0,b_1,\ldots\}}_{\mbox{$\mathcal{B}$ases}}}
\mbox{\hspace{3ex}}\mid\mbox{\hspace{3ex}}
\ptrdomain{S}{\underbracket{\{s_0,s_1,\ldots\}}_{\mbox{$\mathcal{S}$ources}}}
\]

The abstract points first of all concern symbolic expressions, denoted by type $\mathbb{E}$.
As \emph{operators}, symbolic expressions have arithmetic and logical operations, bit-level operations such as sign-extension or masking, and other operations related to x86-64 assembly instructions.
There is a dereference operator $\memread{e}{\mathit{si}}$ -- where $e$ and $\mathit{si}$ are symbolic expressions -- that models reading $\mathit{si}$ bytes from address $e$.
As \emph{operands}, symbolic expressions have \emph{immediate values}, \emph{state parts}, \emph{constants} or \emph{heap-pointers}.
Immediate values are words of fixed size.
State parts can be registers (e.g., \reg{rax}, \reg{edi}, $\ldots$) or flags (\reg{ZF}, \reg{CF}, \ldots).
Constants are values of state parts relative to the initial state. For example, \regz{rax} denotes a constant: the initial value stored in register \reg{rax}.
Constants thus represent initial values of state parts when the current function was called.
A heap-pointer is an expression of the form $\mathtt{alloc}[\mathit{id}]$ and models the return value of an allocation function such as \texttt{malloc}.
The $\mathit{id}$ is an identifier to distinguish different mallocs.
A \emph{constant computation} is a symbolic expression with as operands only immediate values and constants.

\textbf{Pointers with constant computations.}
An expression of the syntax $\ptrdomain{C}{\{c_0,c_1,\ldots\}}$ is based on a non-empty set of constant computations.
It non-deterministically represents any value from the given set.
For example, expression $\ptrdomain{C}{\{\regz{rsp} - 8, \regz{rdi} + 16\}}$ simply represents any of the two constant computations.
This domain is the most concrete.

\textbf{Pointers with bases.}
Expressions of the form $\ptrdomain{B}{\{b_0,b_1,\ldots\}}$ \emph{partly} abstract away from how the pointer is computed.
It keeps track only of positive \emph{addends} in the computation that can be recognized as a basis for a pointer.
The given set is a non-empty set of \emph{bases}.
A base is defined by the following datastructure:
\[
\begin{array}{c}
\mathsf{Base} \equiv \mathsf{StackPointer}~\mathit{f} \mid \mathsf{Global}~\mathit{a} \mid \mathsf{Alloc}~\mathit{id} \mid \mathsf{Symbol}~\mathit{name}
\end{array}
\]
Four types of pointer bases are recognized.
The base may be the stackpointer (in x86-64 this is register \reg{rsp}) pointing to somewhere in the stackframe of a certain function $f$.
The base may be a global address $a$. 
At the binary level, the global address space is pointed to using immediate values and thus $a$ is a 64-bit immediate word.
The base may be the result of some dynamic memory allocation. 
Finally, a base can be a named symbol.
At the binary level, external variables are named symbols with immediate addresses.

Note that any pointer will have at most one base. 
For example, no pointer computation would allow addition of the stack pointer and a heap address.
The datastructure allows a set of bases to allow the abstract pointer to non-deterministically represent different pointers with bases.

\textbf{Pointers with sources.}
Expressions of the form $\ptrdomain{S}{\{s_0,s_1,\ldots\}}$ \emph{fully} abstract away from how the pointer is computed.
Whereas pointers with identifiable bases contain sufficient information at least for roughly establishing a memory designation, a pointer parameterized with sources only concerns what information has been used in the computation of the pointer.
In other words, the pointer has been computed by some expression with operands from the given set of sources.
The following sources are possible:
\[
\begin{array}{c}
\mathsf{Source} \equiv \mathsf{Constant}~\mathit{c} \mid \mathsf{Base}~\mathit{b} \mid \mathsf{Fun}~\mathit{name}
\end{array}
\]
A source can be a constant, a base or the return value of some function.
For example, expression $\ptrdomain{S}{\{\regz{rdi}, \regz{rsi}, \mathsf{Fun}~\mathit{getc} \}}$ indicates a pointer that has been computed using \emph{only} the initial values of registers \reg{rdi} and \reg{rsi} and the return value of function \texttt{getc}.

\begin{example}\label{ex:running1}
Consider the running example in Figure~\ref{fig:running}.
The example allocates memory and performs some memory writes.
Two regions (at addresses \hexa{0x3006} and \hexa{0x3008}) are relative to the initial value of the stackpointer \regz{rsp}.
The abstract pointers corresponding to these memory writes may be represented in $\mathcal{C}$: respectively $\ptrdomain{C}{\{\regz{rsp}-16\}}$ and $\ptrdomain{C}{\{\regz{rsp}-8\}}$.
The memory write at address \hexa{0x3005} occurs on the heap.
Even if $\mathit{offset}$ is some convoluted dynamically computed offset, the abstract pointer will be representable in $\mathcal{B}$: $\ptrdomain{B}{\{\mathsf{Alloc}~\hexa{0x3003}\}}$.
Finally, the memory write at address \hexa{0x3007} writes to the address initially stored in register \reg{rdx} plus the return value of function \texttt{getc}.
The abstract pointer is representable in $\mathcal{S}$: $\ptrdomain{S}{\{\regz{rdx},\mathsf{Fun}~\texttt{getc}\}}$.
\end{example}

\begin{figure*}[tb!]
\hspace{-2ex}\noindent
\begin{tabular}{cc}
  \begin{subfigure}[b]{.42\linewidth}
    \begin{lstlisting}[language={[x86masm]Assembler}, escapeinside=||]
|\hexa{0x3000:}| mov rbp, rsp
|\hexa{0x3001:}| call getc
|\hexa{0x3002:}| mov rcx, rax
|\hexa{0x3003:}| call malloc
|\hexa{0x3004:}| lea rsi, [rbp - 8]
|\hexa{0x3005:}| mov qptr [rax + |$\mathit{offset}$|], rsi
|\hexa{0x3006:}| mov qptr [rsp-16], |\hexa{0x2000}|
|\hexa{0x3007:}| mov qptr [rdx + rcx], rax
|\hexa{0x3008:}| mov dptr [rsi], 0
    \end{lstlisting}
    \caption{x86-64 Assembly}
  \end{subfigure}
&
  \begin{subfigure}[b]{.5\linewidth}
\[
	\begin{array}{lcl}
	\sigma.\reg{rsp} &=& \ptrdomain{C}{\{\regz{rsp}\}}                           \\
  \sigma.\reg{rbp} &=& \ptrdomain{C}{\{\regz{rsp}\}}                           \\
  \sigma.\reg{rcx} &=& \ptrdomain{S}{\{\mathsf{Fun}~\mathit{getc}\}}      \\
 	\sigma.\reg{rax} &=& \ptrdomain{C}{\{\mathtt{alloc}[\hexa{0x3003}]\}}        \\
	\sigma.\reg{rsi} &=& \ptrdomain{C}{\{\regz{rsp} - 8\}}                       \\
\memreadS{\sigma}{\ptrdomain{B}{\{\mathsf{Alloc}~\hexa{0x3003}\}}} &=& \ptrdomain{C}{\{\regz{rsp} - 8\}} \\
\memreadS{\sigma}{\ptrdomain{C}{\{\regz{rsp} - 16\}},8} &=& \ptrdomain{B}{\{\mathsf{Global}~\hexa{0x2000}\}} \\
\memreadS{\sigma}{\ptrdomain{S}{\{\regz{rdx},\mathsf{Fun}~\texttt{getc}\}}} &=& \ptrdomain{C}{\{\mathtt{alloc}[\hexa{0x3003}]\}} \\
\memreadS{\sigma}{\ptrdomain{C}{\{\regz{rsp} - 8\}},4}   &=& \top \\
	\end{array}
\]
    \caption{Abstract State}
    \label{fig:atstate}
  \end{subfigure}
	\end{tabular}
	\caption{Assembly code. For sake of presentation, pseudo code is given on the right. 
$\mathit{offset}$ denotes some dynamically computed offset. Instruction \texttt{\textcolor{blue}{lea}} loads an address into register \reg{rsi} without reading from memory. The binary has a data section with address range \hexa{0x2000} to \hexa{0x2040}.}
  \label{fig:running}
\end{figure*}

\nocompile{
\begin{figure}[tb!]
\begin{lstlisting}[language={[x86masm]Assembler}, escapeinside=||]
|\hexa{0x3000:}|  mov rbp, rsp
|\hexa{0x3001:}|  call getc
|\hexa{0x3002:}|  mov rcx, rax
|\hexa{0x3003:}|  call malloc
|\hexa{0x3004:}|  lea rsi, [rbp - 8]
|\hexa{0x3005:}|  mov qword ptr [rax + |$\mathit{offset}$|], rsi
|\hexa{0x3006:}|  mov qword ptr [rsp-16], |\hexa{0x2000}|
|\hexa{0x3007:}|  mov qword ptr [rdx + rcx], rax
|\hexa{0x3008:}|  mov dword ptr [rsi], 0
\end{lstlisting}
\caption{Assembly code. Instruction addresses are artificially fabricated. $\mathit{offset}$ denotes some dynamically computed offset. Instruction \lstinline[language={[x86masm]Assembler}]{lea} loads an address into register \reg{rsi} without reading from memory. The binary has a data section with address range \hexa{0x2000} to \hexa{0x2040}.}
  \label{fig:running}
\end{figure}
}

\subsection{Separation over Abstract Pointers}

Each of the three domains allows a different type of reasoning over separation.
Intuitively, separation between two pointers implies that \emph{necessarily} -- i.e., for any address represented by the two pointers -- two memory writes commute.
Separation is denoted by the infix operator $\bowtie$.

Consider two pointers $p_0 + \mathit{offset}$ and $p_1$.
If both pointers are defined by constant computations, separation can be decided by an SMT solver.
For the other domains, however, one may argue that separation cannot be guaranteed under any circumstances.
We argue this is too strict.
For example, if one pointer is based on the heap and the other on the local stack, then it is safe to assume that even if $\mathit{offset}$ is chosen nefariously, writes to the two pointers will either cause a segmentation fault or they will commute.
Assuming separation is thus safe.
On the other hand, we argue that if both pointers are based on local stackframes of different functions, then assuming separation is dangerous as a stack overflow may cause the pointers to overlap.

It may be the case that separation is not necessarily true, but \emph{desirable}.
Consider a pointer with source set $\{\regz{rdi}\}$.
The only way for that pointer to overlap with the stackframe of the current function, is if initially -- at the time the current function was called -- register \reg{rdi} contained a pointer to below its own stackframe.
That is invalid source code and can therefore be considered undesirable.
For various use cases, it can be useful to distinguish ``necessarily'' separation from ``desirable'' separation (see Section~\ref{sec:usecases}).

For domains $\mathcal{B}$ and $\mathcal{S}$, whether two abstract pointers are separate is in essence domain specific knowledge on what assumptions can validly be made when dealing with an x86-64 binary.
In Figure~\ref{fig:separation}, this knowledge will be encoded in algebraic definitions.
Depending on ones view, one may easily add or remove cases of separation, and move cases from ``desirable'' to ``necessarily'' or the other way around.

\textbf{Pointers with constant computations.}
For two constant computations, separation can be formulated as a linear programming problem, solvable by an SMT solver~\cite{moura2008z3}.
This requires the size of the regions pointed to, to be known.
Typically, these sizes can syntactically be inferred from the current instruction.
We define function $\mathsf{SMT}[\bowtie]$ as a function that takes as input two constant computations $c_0$ and $c_1$, and two sizes $\mathit{si}_0$ and $\mathit{si}_1$ and returns true if and only if the following SMT problem is unsatisfiable:
$
	c_0 \leq a < c_0 + \mathit{si}_0 \wedge c_1 \leq a < c_1 + \mathit{si}_1
$. 
In words, an address $a$ that is in both regions should be unsatisfiable.
Constant computations typically consist of arithmetic operations containing only $+$, $-$ and $*$, and bit-level operations supported by SMT-theory \texttt{QF\_BV}.

\begin{figure*}[htb]
\[    
\begin{array}{clclcl}
\hlinewd{1pt}
\multicolumn{6}{l}{\text{\textbf{Necessarily:}}}\\
\mbox{\hspace{1em}}
&\mathsf{StackPointer}~f   & \bowtie^\mathsf{B} & \mathsf{Alloc}~\mathit{id} &&\\
&\mathsf{StackPointer}~f   & \bowtie^\mathsf{B} & \mathsf{Global}~\mathit{a} &&\\
&\mathsf{StackPointer}~f   & \bowtie^\mathsf{B} & \mathsf{Symbol}~\mathit{name} &&\\
&\mathsf{Global}~a         & \bowtie^\mathsf{B} & \mathsf{Alloc}~\mathit{id} &&\\
&\mathsf{Alloc}~\mathit{id_0} & \bowtie^\mathsf{B} & \mathsf{Symbol}~\mathit{name} && \\
&\mathsf{Alloc}~\mathit{id_0} & \bowtie^\mathsf{B} & \mathsf{Alloc}~\mathit{id_1} & \mbox{if} & \mathit{id_0} \neq \mathit{id_1} \\\grayhline
&\mathsf{Base}~b_0    & \bowtie^\mathsf{S} & \mathsf{Base}~b_1  & \equiv & b_0 \bowtie^\mathsf{B} b_1 \\
&\mathsf{Constant}~c   & \bowtie^\mathsf{S} & \mathsf{Base~Alloc}~\mathit{id}  & \mbox{if} & \mathit{id} \mbox{~belongs to current function}\\
&\mathsf{Fun}~f  & \bowtie^\mathsf{S} & s_1  & & \\\hline
\multicolumn{6}{l}{\text{\textbf{Desirable:}}}\\
&\mathsf{StackPointer}~f   & \bowtie^\mathsf{B} & \mathsf{StackPointer}~f' & \mbox{if} & f \neq f'\\
&\mathsf{Global}~a   & \bowtie^\mathsf{B} & \mathsf{Global}~a' & \mbox{if} & \mbox{$a$ and $a'$ are from different data sections}\\
&\mathsf{Global}~a   & \bowtie^\mathsf{B} & \mathsf{Symbol}~\mathit{name} & & \\\grayhline
& \mathsf{Constant}~c   & \bowtie^\mathsf{S} & \mathsf{Base~StackPointer}~\mathit{f}  &  \mbox{if} & f \mbox{~is current function}\\\hlinewd{1pt}
\end{array}
\]
\caption{Separation over abstract pointers: the smallest symmetric relation such that the above holds.}
\label{fig:separation}

\end{figure*}

\textbf{Pointers with bases.}
Separation for pointers with bases can be decided algebraically.
Figure~\ref{fig:separation} presents relation~$\bowtie^\mathsf{B}$, i.e., separation over bases.
The local stack frame of any function $f$ is assumed to be separate from the global address space and the heap.
As argued above, separation over different stackframes is desirable, but not necessarily.
The global address space is assumed to be separate from the local address space and the heap.
The global address is not necessarily assumed to be separate from the symbols in the binary, i.e., external variables.
Both are typically constant immediates, within the range of addresses within the binary.
Separation is considered to be desirable.
If two global addresses are based on immediates that fall within the range of different data sections of the binary, then separation is considered desirable.

Two allocations with different ids are assumed to be separate, as different ids ensure these were different calls to \texttt{malloc}.
\deleted{In a similar argument as above, we argue that} A pointer based on \texttt{malloc} \added{should} \deleted{can} not lead to a write overlapping with the region of a different \texttt{malloc}.
\deleted{Even if nefariously an $\mathit{offset}$ is added or subtracted from the pointer, the writes will either cause a fault, or commute.}
\added{Whether this separation should be considered necessary or desirable is debatable.
Heap overflows are common and critical vulnerabilities, and if these are part of the attacker model, then this separation should be considered desirable.}
Note that two pointers based on allocations with the same \textit{id}, are not necessarily \emph{not} separate.


\textbf{Pointers with sources.}
Since pointers with sources are more abstract, there are fewer cases that allow deciding separation necessarily.
For example, consider two pointers with source sets $\{\regz{rdi}\}$ and $\{\regz{rsi}\}$.
Whether these can be considered separate depends on the initial state of the current function, i.e., whether registers \reg{rdi} and \reg{rsi} initially were separate.
However, there still are cases when separation can be decided.
First, when two sources are actually bases, the above relation~$\bowtie^\mathsf{B}$ can used as decision procedure.
Second, consider a pointer computed using constant \regz{rdi} and a newly allocated pointer.
As long as the allocation happened within the current function, separation can safely be assumed: essentially this assumes that the initial state cannot predict where \texttt{malloc} will allocate memory.
Since sources of the form $\mathsf{Fun}~f$ do not represent pointers, separation between these sources and all other sources can be assumed.
Finally, consider a source $\regz{rdi}$, i.e., the initial value of register \reg{rdi} when the current function was called.
It is possible that the caller initialized $\reg{rdi}$ with a pointer to the stackframe of the callee, but we consider it undesirable.

\begin{example}\label{ex:running2}
Revisiting Example~\ref{ex:running1}, we address the memory accesses in order of execution.
First a write happens to pointer $p_0 = \ptrdomain{B}{\{\mathsf{Alloc}~\hexa{0x3003}\}}$.
Then, a write happens to pointer $p_1 = \ptrdomain{C}{\{\regz{rsp}-16\}} = \ptrdomain{B}{\{\mathsf{StackPointer}~\hexa{0x3000}\}}$.
These two regions are necessarily separate for all three domains.
Then, a write happens to pointer $p_2 = \ptrdomain{S}{\{\regz{rdx},\mathsf{Fun}~\texttt{getc}\}}$.
Separation between $p_0$ and $p_2$ follows necessarily for domains $\mathcal{B}$ and $\mathcal{S}$.
However, separation between $p_1$ and $p_2$ is not necessarily true.
It is, however, \emph{desirable}.
Finally, a write happens to pointer $p_3 = \ptrdomain{C}{\{\regz{rsp}-8\}}$.
Similar reasoning applies for separation with $p_0$ and $p_2$.
For domain $\mathcal{C}$, separation between $p_1$ and $p_3$ is decided by an SMT solver.
\end{example}

In similar fashion, enclosure and aliassing are instantiated.
The join is defined as set-union. 
Instantiating abstract semantics and concretization is straightforward. 

\section{Intraprocedural Pointer Analysis}\label{sec:intra}

Algorithmically, intraprocedural pointer analysis can be achieved using a standard abstract interpretation algorithm~\cite{cousot1977abstract,horwitz1987efficient}.
Essentially, starting at some initial abstract state with the instruction pointer set to the entry point of the binary (or a function of interest), symbolic execution traverses the assembly instructions step-by-step using function $\mkabstract{\mbox{step}}$.
It runs until a return statement, or an exiting function call.
Whenever an instruction address is visited twice, the current state $\sigma_\text{curr}$ is compared to the state $\sigma_\text{stored}$ belonging to the previous visit.
If state $\sigma_\text{stored}$ is more abstract then state $\sigma_\text{curr}$, exploration can stop.
Otherwise, the two states must be \emph{joined} to a state $\sigma_\text{join}$. 
That state is stored and exploration continues.
Effectively, this algorithm provides a (partial) mapping from instruction addresses to stored symbolic states.
A post-state $\sigma_\text{post}$ is computed by taking the supremum of all terminal non-exiting states.
Thus, we define \emph{state} and a \emph{symbolic step function} on top of the generic constituents.


\textbf{Abstract States.}~
A state stores values for registers, flags and memory.
An abstract \emph{state part} is defined by the following datatype:
\[
\begin{array}{c}
\mkabstract{\mathbb{SP}} \equiv \mathsf{Register}~\reg{r} \mid \mathsf{Flag}~\reg{f} \mid \mathsf{Memory}~(\mkabstract{\mathbb{V}} \times \mathbb{N}?)
\end{array}
\]
A memory statepart is represented as a region, a tuple with an abstract pointer and optionally a size.
Whenever clear from context, we will omit the constructors.
For example, $\mathit{sp} = \reg{rax}$ simply means that statepart $\mathit{sp}$ equals register $\reg{rax}$.
An abstract state is a partial mapping from stateparts to abstract values:
$
	\mkabstract{\mathbb{S}} \equiv  \mkabstract{\mathbb{SP}} \rightharpoonup \mkabstract{\mathbb{V}}
$.
Notation $\sigma.\reg{r}$ denotes the current value stored in register~\reg{r}.
Memory partitions the address space into \emph{separate} abstract regions, and assigns stored values to each of these regions.
Notation $\memreadS{\sigma}{r}$ denotes the value stored in region $r$, where the size is omitted when not available.


\begin{example}\label{ex:state}
Revisit the running example in Figure~\ref{fig:running}.
Based on the separation relations decided in Example~\ref{ex:running2}, the state $\sigma$ in Figure~\ref{fig:atstate} can model the state after execution of the assembly snippet.
Note that this state is based on \emph{desirable} separation, as otherwise regions would need to get joined.
Intuitively, the following claims are modeled by this state:
\begin{itemize}
\item Any region based on a heap-pointer allocated at line \hexa{0x3003} contains no pointers other than the pointer $\regz{rsp} - 8$.
\item The 8-byte region at address $\regz{rsp} - 16$ stores a value that may point to the global data section.
\item A memory write has occurred that is determined by the initial value of register \reg{rdx} as well as input provided via function \texttt{getc}. That memory write was assumed to be separate from all others.
\item If the value stored in region $\mkregion{\regz{rsp} - 8,4}$ would be dereferenced, then no information on its designation in memory is known. 
\end{itemize}
\end{example}

The claims in the above example can be formalized by \emph{concretization} of abstract states to concrete ones.
Concretization of state parts is provided by function $\gamma_\mathbb{SP}$ of type $\mkabstract{\mathbb{SP}} \mapsto 2^{\mathbb{SP}}$:
\[
\begin{array}{lcl}
\gamma_\mathbb{SP}(\mathsf{Register}~\reg{r}) &\equiv& \{~ \mathsf{Register}~\reg{r}~\}\\
\gamma_\mathbb{SP}(\mathsf{Flag}~\reg{f}) &\equiv& \{~\mathsf{Flag}~\reg{f}~\}\\
\gamma_\mathbb{SP}(\mathsf{Memory}~\mkabstract{r}) &\equiv& \{~ \mathsf{Memory}~r \mid r \in \gamma_\mathbb{R}(\mkabstract{r})~\}
\end{array}
\]

\begin{definition}\label{def:concrete_states}
The \emph{concretization} of states is function $\gamma_\mathbb{S}$ of type $\mkabstract{\mathbb{S}} \mapsto 2^{\mathbb{S}}$, and is defined as:
\[
\begin{array}{l}
\gamma_\mathbb{S}(\sigma) \equiv 
\{ s \mid \forall \mkabstract{\mathit{sp}} \cdot \sigma.\mkabstract{\mathit{sp}} = \mkabstract{v} \implies  \\
\mbox{\hspace{15ex}}\forall \mathit{sp} \in \gamma_\mathbb{SP}(\mkabstract{\mathit{sp}}) \cdot\exists v \in \gamma_\mathbb{V}(\mkabstract{v}) \cdot s.\mathit{sp} = v \}
\end{array}
\]
\end{definition}
In words, abstract state $\sigma$ is mapped to any concrete state $s$ in which both all stateparts and all stored values have been concretized.

\textbf{Abstract Step Function.}~
Every assembly instruction can be written as a sequence of micro-instructions of the following form:
\[
	\mathtt{dst} \coloneqq f(\mathtt{in}_0, \mathtt{in}_1, \ldots)
\]
A single destination -- be it a register, a memory region or a flag -- gets assigned the return value of a pure operation $f$ that is based on zero or more input-operands (registers, memory regions, flags, or immediate values).
The operation typically corresponds to a mnemonic, e.g., \texttt{add} or \texttt{imul}.
In x86-64, many (but not all) instructions have one destination operand (register or memory) and may set a list of flags.
Typically, it cannot be the case that both destination and sources are memory.
We here make no assumptions, and generalize the symbolic step function over any sequence of micro-instructions.
Transformation from basic assembly to sequences of micro-instructions can be done, e.g., using the Ghidra decompiler 
which translates assembly instructions of various architectures into \emph{low P-code}~\cite{rohleder2019hands}.

The semantics of executing a micro-instruction $\mu$ consists of resolving the input-operands, applying operation $f$, and writing the produced value to the state.
This is implemented in function $\mkabstract{\mbox{step}}(\mu,\sigma)$.
The abstract semantics $\mkabstract{S}$ are used to apply operation $f$ on abstract values.
Functions $\mathsf{read}(r, \sigma)$ and $\mathsf{write}(r, a, \sigma)$ are defined that take care of memory accesses in abstract states.

\textbf{Correctness and Termination.~}
Correctness -- the abstract semantics overapproximate the concrete semantics -- has been defined at the start of Section~\ref{sec:overview} and has been formally proven correct in Isabelle/HOL.
Note that it may be the case that paths are unexplored due to unresolved indirections.

Termination has also been proven, but it requires the proof obligation that there does not exist an infinite chain of different states such that $\sigma_0 \sqcup \sigma_1 \sqcup \ldots$.
We thus ensure termination by putting bounds on the sizes of the sets of the abstract pointers.
These bounds are chosen manually.
They do not affect correctness.
Making them larger makes the pointers more precise, but increases running times.
For bounds $\mathcal{C}$, $\mathcal{B}$, and $\mathcal{S}$  the bounds are resp. 10, 5 and 250.
Any pointer with more elements is shifted to $\mkabstract{\top}$.

\section{Use Cases}\label{sec:usecases}

We here discuss several use cases with small pedagogical examples.

\textbf{Integration into Disassembly:}
A fundamental problem in disassembly is resolving indirections.
Typically, indirect jumps can be analyzed through \emph{intra}procedural analysis. For example, they are the result of reading a jump table induced by a \texttt{switch} statement.
For indirect calls, however, often \emph{inter}procedural analysis is necessary.
An indirect call is often the result of a function callback.
These typically happen across function boundaries: a function pointer is passed from function to function until it is called.
This scenario mandates context-sensitive interprocedural analysis where the context provides sufficient information to resolve indirect calls.


\begin{figure}[htb]
\begin{tabular}{c}
  \begin{subfigure}[b]{\linewidth}
    \begin{lstlisting}[language={[x86masm]Assembler}, escapeinside=||, columns=fixed]
|Function $f$:|
|\hexa{0x6000:}|  mov qword ptr [0x2010], 0x6050
|\hexa{0x6001:}|  call 0x6500
|\hexa{0x6002:}|  call exit
|\ldots|
|\hexa{0x6050:}|  push rbp
|\hexa{0x6051:}|  |\ldots|
    \end{lstlisting}
  \end{subfigure}
 \\ 
  \begin{subfigure}[b]{\linewidth}
    \begin{lstlisting}[language={[x86masm]Assembler}, escapeinside=||, columns=fixed]
|\ldots|
|Function $g$:|
|\hexa{0x6500:}|  mov rax, qword ptr [0x2010]
|\hexa{0x6501:}|  call rax
|\hexa{0x6502:}|  ret
    \end{lstlisting}
  \end{subfigure}
\end{tabular}
\caption{Example of Indirect Call}
\label{fig:example_indirect_call}
\end{figure}

Figure~\ref{fig:example_indirect_call} provides an example.
Function pointer \hexa{0x6050} is written to a global variable (at address \hexa{0x2010}).
Function $h$ will be analyzed in a context where the pre-state provides:
\[
	\memreadS{\sigma_\text{pre}}{\ptrdomain{B}{\{\mathsf{Global}~\hexa{0x2010}\}}} = \ptrdomain{C}{\{\hexa{0x6050}\}}
\]
When symbolically executing the indirect call at Line~\hexa{0x6501}, it is overapproximatively known that any pointer stored in the global address space based on address \hexa{0x2010} points to instruction address \hexa{0x6050}.
The recursive traversal can therefore consider \hexa{0x6050} as a reachable instruction address.
Note that in a stripped binary, functions are not delineated: it is not known what addresses are function entries.
Without context, the indirect call at Line~\hexa{0x6501} cannot be resolved and the entire function at entry \hexa{0x6050} would have been missed.
\added{Section~\ref{subsec:perinst} provides data on how many indirections can be resolved in practice.}

\textbf{Preliminary to Decompilation:}
At the assembly level, there are no variables.
Bottom-up analysis, such as decompilation, has no ground truth as to what regions in memory constituted variables.
A variable in source code typically is compiled to a memory region, with the characteristic that this region is separate from any memory write not to the same variable.
As example consider Figure~\ref{fig:decompilation}.
It shows assembly with some hypothetical pointer analysis result on the right.
Based on the derived pointers, a decompilation tool can assign a variable $\mathsf{x}$ to lines \hexa{0x7000} and \hexa{0x7002}.
\begin{figure*}[htb]
\centering
\begin{tabular}{c|c|c}
  \begin{subfigure}[b]{.45\linewidth}
    \begin{lstlisting}[language={[x86masm]Assembler}, escapeinside=||]
|\hexa{0x7000:}|  mov qword ptr [rdi], 42
|\hexa{0x7001:}|  mov qword ptr [rsi], 43
|\hexa{0x7002:}|  mov qword ptr [rbp-40], 44
    \end{lstlisting}
    \caption{x86-64 Assembly}
  \end{subfigure}
	&
  \begin{subfigure}[b]{.2\linewidth}
    \begin{tabular}[b]{lcl}
    $\mathbf{x}$                 &  $\coloneqq$ & 42\\
    *\reg{rsi}                   &  $\coloneqq$ & 43\\
    $\mathbf{x}$                 &  $\coloneqq$ & 44\\
    \end{tabular}
    \caption{Decompiled}
  \end{subfigure}
	&
  \begin{subfigure}[b]{.25\linewidth}
    \begin{tabular}[b]{lcl}
    \reg{rdi}                 &  $=$ & $\ptrdomain{C}{\{\regz{rsp} - 48\}}$\\
    \reg{rsi}                 &  $=$ & $\ptrdomain{B}{\{\mathsf{Alloc}~\mathit{id}\}}$\\
    \reg{rbp}                 &  $=$ & $\ptrdomain{C}{\{\regz{rsp} - 8\}}$\\
    \end{tabular}
    \caption{Abstract Pointers}
  \end{subfigure}
\end{tabular}
\caption{Assembly code, decompiled to code with variables based on abstract pointers.}
\label{fig:decompilation}
\end{figure*}

\textbf{Bottom-up Dataflow Analysis:}
The post-state produced by analysis of a function provides overapproximative insight into what parts of the state are modified by a function.
One application of this, is that it can be used to verify whether a certain function abides by a calling convention.
The calling convention designates certain registers to be non-volatile, i.e., they must be preserved by a function.
This can be observed from the post-state directly. 
If the post-state provides that, for a register \reg{r}:
$
\sigma_\text{post}.\reg{r} = \ptrdomain{C}{\{\regz{r}\}}
$,
then this indicates that register \reg{r} has been properly preserved.
Many compilers use a push/pop pattern to achieve such preservation: a register value is initially pushed to the local stack, and popped just before return.
Calling convention adherence requires abstract pointer analysis to be sufficiently precise throughout symbolic execution that the local region into which a register value is pushed is not overwritten.
A push/pop pattern is not necessary though: our approach is transparent to the means of register preservation a function may use.

We provide an additional example demonstrating this overapproximative form of pointer analysis can be used for live variable analysis.
Consider a function returning with the state in Example~\ref{ex:state}.
The pointer allocated at line \hexa{0x3003} is stored in the return register $\reg{rax}$ and can therefore considered to be live.
However, if the state overapproximately indicates that the pointer is not returned and not written to global memory, the pointer can be considered as ``not live'' after return of the function.

\added{
\textbf{Finding suspect patterns for automated exploit generation:}
Pointer analysis can enhance real-world downstream analyses.
As an example, we consider automated exploit generation: automatically finding bugs and generating working exploits~\cite{avgerinos2014automatic,you2017semfuzz}.
One of the many challenges in this field, is to deal with state space explosion.
Our pointer analysis can be used to find patterns in the binary that may lead to vulnerabilities, thereby pruning the state space to be explored.
}

\added{
As a concrete case study, we can enumerate all instructions in a binary that do a function call to an external function, \emph{and} that pass a pointer to the current stackframe as parameter to that function.
This is a suspect pattern, as the external function has been given the opportunity to overwrite the return address.
We have applied our pointer analysis to the \texttt{ret2win} challenge provide by ROP emporium}\footnote{\url{https://ropemporium.com}}.
\added{
The above heuristic finds an instruction at address \hexa{0x400701}: \texttt{\textcolor{blue}{call memset}} where register \texttt{\textcolor{blue}{rdi}} (the first parameter under the System V ABI calling convention) contains a local pointer. Indeed, the example is (purposefully) exploitable, and the exploit leverages exactly this particular instruction.}


\section{Evaluation}\label{sec:evaluation}

In addition to the formal proofs of the soundness of our approach, we provide a prototype implementation and conducted a series of experimental studies to evaluate \textit{soundness} and \textit{preciseness}. 
We run 1.) a \emph{comparative} evaluation against the state-of-the-art, and 2.) a more in-detail evaluation \emph{per instantiation}.

\subsection{Comparative Evaluation}

The closest related works to our binary pointer analyzer are BinPointer~\cite{kim2022binpointer} and BPA~\cite{kim2021refining}. Unfortunately, even after contacting the authors of the papers, either their source code was not available or the code was not runnable on our examples. Therefore, to stay fair in our comparison, we used the exact same dataset of binaries that BinPointer and BPA used (SPEC\_2006\_V1) and we employed the same definition of soundness and precision they used in their paper.

We map all abstract pointers to a subset of $\{ L, G, H \}$ (for: local, global, heap).
Domains that are provably within the current stackframe are mapped to $\{L\}$, domains that are relative to the stackpointer but that may possibly be above the current stack frame (e.g., point to the stackframe of a caller) are mapped to $\{L,H\}$.
Global bases are mapped to $\{G\}$.
Top is mapped to $\{ L, G, H \}$, the rest is considered heap and mapped to $\{H\}$.
This produces function $\mathsf{PA}$ that maps instruction addresses to observations.

To assess the ground truth, we built an instrumentation tool based on the dynamic binary instrumentation tool PIN~\cite{luk2005pin}.
Each memory write observed at run-time to some address $a$ is mapped to an element of $\{ L, G, H \}$.
If $\regz{rsp} - 0x10000 \leq a \leq \regz{rsp}$, address $a$ is considered local (here $\regz{rsp}$ denotes the value of the stackpointer when the current function was called).
We cannot know how large the stackframe is, and overapproximate the size with the constant $0x10000$ for sake of observation.
If address $a$ is a memory address covered by any of the sections where the binary is located at run-time, it is considered global.
In all other cases, address $a$ is considered heap.
By running multiple executions, each memory write is mapped to a subset of $\{ L, G, H \}$ as well.
This produces function $\mathsf{GT}$ that maps instruction addresses to observations.
The total set of (instruction addresses of) observed memory writes is denoted with $\mathsf{W}$.

\textit{Soundness:} We consider a static pointer analysis to be sound if its results support the ground truth, i.e., the observations of PIN. 
In other words, the pointer analysis must predict a superset of the ground truth.
Soundness is computed using the \emph{recall}: the percentage of supported memory writes to the total number.
If the recall is $100\%$, the pointer analysis is sound. 
    
\textit{Precision:} Soundness does not imply usefulness. For example, if the pointer analysis returns $\top$ for all the instructions, it will be consider sound but useless. We therefore measure \emph{precision} as well. In words, precision measures the specificity of the returned pointer domains. 
The precision is computed as the total average of the percentage of domains that $\mathsf{PA}$ overapproximated but were not observed by $\mathsf{GT}$.
\[
\begin{array}{lcl@{\hspace{5ex}}@{\hspace{5ex}}lcl}
\mathsf{recall} &\equiv& 100 * \frac{|\{ a \in \mathsf{W} \cdot \mathsf{GT}(a) \subseteq \mathsf{PA}(a)\}|}{|W|} 
& \\
\mathsf{precision} & \equiv& 100 * \mathsf{avg}_{a \in \mathsf{W}}(1 - \frac{|\mathsf{PA}(a) \setminus \mathsf{GT(a)}|}{3})
\end{array}
\]

Table~\ref{tbl:spec_experiment} shows results.
The authors of BinPointer have reported that the recall of their approach on the SPEC dataset is 100\%. Since we also observed the same recall for our approach on this dataset, we do not include it in the table.
As the results show, when it comes to heap and global memory accesses, our approach achieves over a $29\%$ and $15\%$ higher precision compared to BinPointer and BPA. For the local memory accesses, our approach shows almost similar precision compared to that of BinPointer.   
All other larger case studies reported in~\cite{kim2022binpointer} are not publicly available. 
Interpreting and comparing their published results subjectively, we conclude that our work achieves at least similar precision.
The largest reported binary for BinPointer is 161K instructions.
In the next section, we show that our tool scales upto 507K instructions, and we thus argue we are at least as scalable as well.

\begin{table*}[tbh!]
    \centering
    \begin{tabular}{ccccccccccc}
    \hline
        \multirow{2}{*}{Binary} & \multirow{2}{*}{\#instructions} & \multicolumn{3}{c}{Local} & \multicolumn{3}{c}{Global} & \multicolumn{3}{c}{Heap}  \\
        \cmidrule(r){3-5}\cmidrule(lr){6-8}\cmidrule(l){9-11} 
        & & BPA & BinPointer & This paper & BPA & BinPointer & This paper & BPA & BinPointer & This paper \\
        \hline
        mcf & 2.4K & 26.3 & 100.0 & 100.0  & 27.0 & 85.7 & 91.5  & N/A & N/A & 57.8\\
        
        lbm & 2.2K & 22.3 & 99.5 & 100.0  & 73.1 & 100.0 & 66.7  & N/A & N/A & 33.3\\
        
        libquantum & 9.6K & 47.9 & 100.0 &  100.0 & 100.0 & 100.0 & 100.0  & 6.9 & 6.9 & 50.0\\
        
        bzip2 & 11K & 16.9 & 93.2 & 89.2  & 51.7 & 51.7 &  100 & 3.6 & 21.8 & 40 \\
        
        sjeng & 22K & 32.7 & 97.5 & 79.1  & 55.1 & 55.6 & 99.8  & N/A & N/A & 56.1 \\
        
        milc & 23K & 49.4 & 99.4 &  93.3 & 81.2 & 88.9 & 95.2  & 23.7 & 23.7 & 51.1\\
        
        hmmer & 60K & 38.0 & 99.9 & 98.9  & 76.4 & 76.4 &  100.0 & 7.6 & 11.5 & 65.4 \\

       h264ref & 100K & 35.3 & 97.3 &  97.1 & 6.2 & 65.5 &  90.6 & 24.2 & 40.8 & 47.6 \\
        \hline
         Average &  & 33.6 & 98.4 & 94.7 & 58.8 & 78.0 & 93.0 & 13.2 & 20.9 & 50.2\\
       
    \end{tabular}
    \caption{Comparison study of the precision (\%) of BPA, BinPointer, and this paper on the SPEC dataset.}
    \label{tbl:spec_experiment}
\end{table*}

\subsection{Per-instantiation Evaluation}\label{subsec:perinst}

We also evaluate per instantiation soundness and precision.
As case studies, we consider a set of binaries covering over 1M instructions (see Table~\ref{tbl:evaluation3}).
Overall recall is 100\%.

We discuss precision in more detail (see Figure~\ref{fig:comparison2}).
Instantiation $\mathcal{C}$ is always 100\% precise -- as no abstraction is applied -- unless it assigns $\top$ to a pointer.
It may assign, e.g., $\ptrdomain{C}{\{\regz{rsp} - 8, \hexa{0x3000}\}}$, in which case on different paths the pointer can be local or global (0.79\% of the overall memory writes).
Overall, instantiation $\mathcal{C}$ assigns a non-$\top$ pointer to 85.22\% of all memory writes.

\begin{figure*}[htb]
\centering
\begin{tikzpicture}[scale=0.85,transform shape,venncircle/.style={draw, circle, minimum size=3cm,align=center}] 
    \node[venncircle] at (0, 0) (circle1) {};    
    \node[venncircle] at (-1cm, -1.75cm) (circle2) {};    
    \node[venncircle] at (1cm, -1.75cm) (circle3) {};
    \node[draw,rectangle,minimum size=5.75cm] at (0cm,-1cm) {};

    \node at (0,1.25cm) {\textbf{L}};
    \node at (-1cm,-3cm) {\textbf{G}};
    \node at (1cm,-3cm) {\textbf{H}};

    \node at (-1.75cm,-2.45cm) (G) {14.14};
    \node at (0cm,-2cm) (GH) {0.02};
    \node at (1.75cm,-2.45cm) (H) {26.38};
    \node at (0cm,0.75cm) (L) {43.83};
    \node at (-1cm,-0.55cm) (LG) {0.79};
    \node at (0cm,-1.2cm) (LGH) {0.06};

    \node at (0cm,-3.55cm) (Top) {$\mathbf{\top}:~14.78$};
    \node at (0,2.25cm) {$\mathcal{C}$};

    \node[venncircle] at (6, 0) (circle1) {};    
    \node[venncircle] at (5cm, -1.75cm) (circle2) {};    
    \node[venncircle] at (7cm, -1.75cm) (circle3) {};
    \node[draw,rectangle,minimum size=5.75cm] at (6cm,-1cm) {};

    \node at (6,1.25cm) {\textbf{L}};
    \node at (5cm,-3cm) {\textbf{G}};
    \node at (7cm,-3cm) {\textbf{H}};

    \node at (4.25cm,-2.45cm) (G) {14.45};
    \node at (6cm,-2cm) (GH) {0.02};
    \node at (7.75cm,-2.45cm) (H) {26.70};
    \node at (6cm,0.75cm) (L) {44.11};
    \node at (5cm,-0.55cm) (LG) {0.82};
    \node at (6cm,-1.2cm) (LGH) {0.06};

    \node at (6cm,-3.55cm) (Top) {$\mathbf{\top}:~13.83$};
    \node at (6cm,2.25cm) {$\mathcal{B}$};

    \node[venncircle] at (12cm, 0) (circle1) {};    
    \node[venncircle] at (11cm, -1.75cm) (circle2) {};    
    \node[venncircle] at (13cm, -1.75cm) (circle3) {};
    \node[draw,rectangle,minimum size=5.75cm] at (12cm,-1cm) {};

    \node at (12cm,1.25cm) {\textbf{L}};
    \node at (11cm,-3cm) {\textbf{G}};
    \node at (13cm,-3cm) {\textbf{H}};

    \node at (10.25cm,-2.45cm) (G) {14.45};
    \node at (12cm,-2cm) (GH) {0.33};
    \node at (13.75cm,-2.45cm) (H) {40.11};
    \node at (12cm,0.75cm) (L) {44.13};
    \node at (11cm,-0.55cm) (LG) {0.82};
    \node at (12cm,-1.2cm) (LGH) {0.07};
    \node at (13cm,-0.55cm) (LH) {0.07};

    \node at (12cm,-3.55cm) (Top) {$\mathbf{\top}:~0$};
    \node at (12cm,2.25cm) {$\mathcal{S}$};    
\end{tikzpicture} 
\caption{Precision results per instantiation. The numbers are percentages. For example, instantiation $\mathcal{C}$ assigned a local pointer to 43.83\% of all memory accesses, and to 0.79\% a pointer set with both a local and a global pointer.}
\label{fig:comparison2}

\end{figure*}
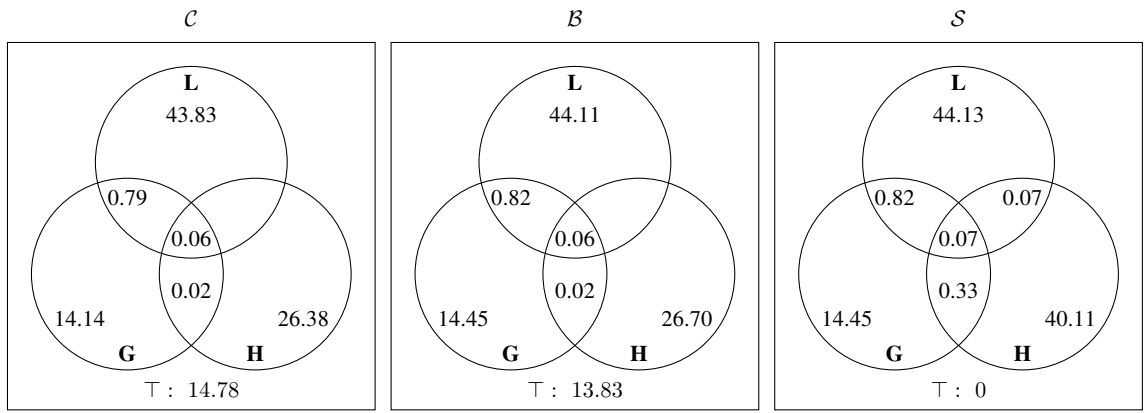

Instantiation $\mathcal{B}$, then, only marginally improves that number.
It assigns pointers to 0.95\% more memory writes, and for those only the base of the pointer is known.
However, the value of this instantiation fluctuates per binary: for some it did not improve on $\mathcal{C}$ at all, but for others up to 5\% more memory writes got assigned a non-$\top$ pointer.
These are typically the hard cases, e.g., a local pointer to an array for which only a base could be established.

Instantiation $\mathcal{S}$, finally, assigns a non-$\top$ pointer to all memory writes.
Figure~\ref{fig:comparison2} shows that this mostly concerns heap-pointers.
From this we can conclude that the cap on the number of sources is never hit, since otherwise this instantiation would assign $\top$.
For those writes where both $\mathcal{C}$ and $\mathcal{B}$ assign~$\top$, instantiation $\mathcal{S}$ at least provides us with information on which sources (e.g., function inputs or \texttt{malloc}s) were used to compute the pointer value.

A holistic way of looking at precision, is 1.) to see whether the pointer analysis is sufficiently precise such that for each function all its pointers are assigned a domain separate from the top of the stack frame (where the return address is stored), and 2.) that it enables resolving indirections.
Table~\ref{tbl:evaluation3} shows that for 94.2\% of all functions this was the case (\textbf{OK}).
For 5.41\% of the functions, not all indirections could be resolved (\textbf{UN}), but all resolved paths ending in a return were \textbf{OK}.
Finally, for 0.39\% of the functions, our pointer analysis was not sufficiently precise to show that the return address was not overwritten (\textbf{ERR}).
This typically happens for functions with complex stackpointer manipulation, such as stack probing or dynamic stack allocation.

\begin{table}[thb]
    \centering
    \begin{tabular}{lllllll}
         \hline\hline
         \textbf{Binary} & \textbf{\#instrs} & \textbf{Time (m:ss)} & \multicolumn{3}{c}{\textbf{\#functions}}\\
				 &&& \textbf{OK} & \textbf{UN} & \textbf{ERR} \\
         \hline
         \texttt{du}        & 30K  & 0:09   & 173  & 7   & 0 \\
         \texttt{gzip}      & 14K  & 0:04   & 101  & 5   & 0 \\
         \texttt{host}      & 12K  & 0:04   & 62   & 8   & 2 \\
         \texttt{sha512sum} & 10K  & 0:05   & 36   & 3   & 0 \\
         \texttt{sort}      & 18K  & 0:08   & 146  & 8   & 0 \\
         \texttt{spec/*}    & 150K & 1:22   & 942  & 41  & 11 \\
         \texttt{sqlite3}   & 319K & 1:33   & 1687 & 144 & 8 \\
         \texttt{ssh}       & 124K & 1:01   & 523  & 32  & 8 \\
         \texttt{tar}       & 91K  & 0:15   & 300  & 17  & 0 \\
         \texttt{vim}       & 507K & 4:26   & 2922 & 86  & 1 \\
         \texttt{wc}        & 6K   & 0:01   & 46   & 4   & 0 \\
         \texttt{wget2}     & 61K  & 0:10   & 578  & 81  & 2 \\
         \texttt{xxd}       & 2K   & 0:01   & 13   & 0   & 0 \\
         \texttt{zip}       & 24K  & 1:09   & 103  & 2   & 0 \\
         \hline
         \textbf{Total} & 1.4M & & 94.2\% & 5.41\% & 0.39\%\\
         \hline\hline
    \end{tabular}
    \caption{Evaluation on an Apple M1 Pro with 32 GB of memory. Memory usage was at most roughly 15GB for \texttt{vim}.}
    \label{tbl:evaluation3}
\end{table}



\section{Conclusion}\label{sec:conclusion}

This paper presented an approach to formally proven correct binary-level pointer analysis, that aims to assign a designation to each memory write in a binary.
The designations are provably overapproximative: the write provably cannot occur to any region in memory outside of its designation.
Evaluation confirms this soundness, and shows that precision is comparable to or improves upon the state-of-the-art.



Many existing approaches to binary analysis, whether it is disassembly, decompilation, binary patching or security analysis, are unsound.
State-of-the-art tools apply heuristics, incorporate best practices, and are generally based on extensive human-in-the-loop knowledge.
Decompilation becomes a form of art rather than an algorithm.
We envision an overapproximative --~provably sound~-- approach as an alternative.
This requires provably sound disassembly, control flow reconstruction, decompilation, and type inference, to begin with.
At the heart of all of these overapproximative techniques lies a proper understanding of the semantics of each individual instruction.
Binary-level pointer analysis aims to aid these future endeavors in overapproximative binary analysis, by indicating what the effect of each memory write in a binary can be.

\section*{Data-Availability Statement}

The implementation of the pointer analysis in this paper, all case studies, and the formal Isabelle/HOL proofs are available anonymously at:
\url{https://doi.org/10.5281/zenodo.14223108}.
\everymath{\color{black}}
\everydisplay{\color{black}}
\color{black}


\section*{Acknowledgments}
We would like to thank the anonymous reviewers for their insightful comments and suggestions, which helped to greatly improve the paper.
 
This work is supported by the Defense Advanced Research Projects Agency (DARPA) and Naval Information Warfare Center Pacific (NIWC Pacific) under Contract No. N66001-21-C-4028, by DARPA and the Army Contracting Command Aberdeen Proving Grounds (ACC-APG) under Prime Contract No. W912CG23C0024, and by DARPA under Prime Contract No. HR001124C0492.



\bibliographystyle{IEEEtranS}
\bibliography{IEEEabrv,ref}

\end{document}